\documentclass[11pt]{article}
\usepackage[ruled,vlined,linesnumbered]{algorithm2e}
\usepackage{amsfonts}
\usepackage{amssymb}
\usepackage{amstext}
\usepackage{amsmath}
\usepackage{amsthm}
\usepackage{bbm}
\usepackage{enumitem}
\usepackage{geometry}
\usepackage{mathrsfs}
\usepackage{tikz}
\usepackage{float}
\usetikzlibrary{decorations.pathreplacing,angles,quotes}

\newtheorem{lemma}{Lemma}
\newtheorem{theorem}{Theorem}
\newtheorem{corollary}{Corollary}

\newtheorem{claim}{Claim}
\theoremstyle{definition}
\newtheorem{definition}{Definition}

\newcommand{\R}{\mathbb R}

\newcommand{\size}[1]{\ensuremath{\left|#1\right|}}

\newcommand{\set}[1]{\ensuremath{\left\{#1\right\}}}

\tikzstyle{vertex}=[circle,draw,align=center]
\tikzstyle{node}=[circle,draw,fill=white!100,inner sep=0pt,minimum size=0.15cm,align=center]
\tikzstyle{weight} = [font=\small, black]
\tikzstyle{edge} = [draw,-]
\tikzstyle{selected edge} = [draw,line width=3pt,-,gray!70]
\tikzstyle{matched edge} = [draw,line width=3pt,-]
\tikzstyle{dashed edge} = [draw,dashed]
\tikzstyle{selected vertex} = [node, fill=gray!70, text=white]

\tikzstyle{remove selected vertex} = [vertex, fill=white!100]
\tikzstyle{blue colored edge} = [draw,line width=3pt,-,blue!80]
\tikzstyle{blue colored curved edge left} = [draw,line width=3pt,-,blue!80,bend left]
\tikzstyle{blue colored curved edge right} = [draw,line width=3pt,-,blue!80,bend right]
\tikzstyle{green colored edge} = [draw,line width=3pt,-,green!80]
\tikzstyle{remove selected edge} = [draw,line width=4pt,-,white!100]
\tikzstyle{remove matched edge} = [draw,line width=4pt,-,white!80,bend right]

\SetArgSty{textup}
\DontPrintSemicolon

\newlength{\rad}
\title{\textbf{An Efficient Characterization of Submodular\\Spanning Tree Games\footnote{A preliminary version of this paper appeared in Proceedings of the 20th Conference on Integer Programming and Combinatorial Optimization, {IPCO} 2019.}}}

\author{Zhuan Khye Koh\footnote{Department of Mathematics, London School of Economics, UK (z.koh3@lse.ac.uk). This work was done while the author was at the University of Waterloo.} \and Laura Sanit\`{a}\footnote{Department of Combinatorics and Optimization, University of Waterloo, Canada (lsanita@uwaterloo.ca).}}

\date{\vspace{-3ex}}

\begin{document}

\maketitle

\begin{abstract}
\emph{Cooperative games} are an important class of problems in game theory, where the goal is to distribute a value among a set of players who are allowed to cooperate by forming coalitions. An outcome of the game is given by an allocation vector that assigns a value share to each player. A crucial aspect of such games is \emph{submodularity} (or \emph{convexity}). Indeed, convex instances of cooperative games exhibit several nice properties, e.g. regarding the existence and computation of allocations realizing some of the most important solution concepts proposed in the literature. For this reason, a relevant question is whether one can give a polynomial time characterization of submodular instances, for prominent cooperative games that are in general non-convex.

In this paper, we focus on a fundamental and widely studied cooperative game, namely \emph{the spanning tree game}. An efficient recognition of submodular instances of this game was not known so far, and explicitly mentioned as an open question in the literature. We here settle this open problem by giving a polynomial time characterization of submodular spanning tree games.
\end{abstract}

\section{Introduction}
\emph{Cooperative games} are among the most studied classes of problems in game theory, with plenty of applications in economics, mathematics, and computer science. In such games, the goal is to distribute cost (or revenue) among a set of participants, usually called \emph{players}, who are allowed to cooperate. Formally, we are given a set of players $N$, and a characteristic function $\nu: 2^N \rightarrow \R$, with $\nu(\emptyset) = 0$. Here, $\nu(S)$ represents the cost paid (revenue received) by the subset of players $S$ if they choose to form a coalition. An outcome of the game is given by an \emph{allocation} $y \in \R^N$ such that $\sum_{v \in N} y_v = \nu(N)$, which assigns a cost (revenue) share to each player. Of course, there are a number of criteria for evaluating how ``good'' an allocation is, such as \emph{stability}, \emph{fairness}, and so on.

Probably the most popular solution concept for cooperative games is the \emph{core}. It is the set of stable outcomes where no subset of players has an incentive to form a coalition to deviate. In a cooperative cost game, this translates naturally to the following constraint: $\sum_{v \in S} y_v \leq \nu(S)$, for all $S \subseteq N$. Intuitively, if this constraint is violated for some set $S$, the total cost currently paid by the players in $S$ is more than the total cost $\nu(S)$ they would have to pay if they form a coalition -- this incentivizes these players to deviate from the current allocation. Besides the core, there are several other crucial solution concepts which have been defined in the literature, e.g. the \emph{Shapley value}, the \emph{nucleolus}, the \emph{kernel}, the \emph{bargaining set}, and the \emph{von Neumann-Morgenstern solution set} (we refer to \cite{book/synthesis/ChalkiadakisEW11} for details). Many fundamental questions involving such solution concepts have been investigated in the past few decades: Which cooperative game instances admit an allocation realizing a particular solution concept? Can we efficiently compute it? Can we test whether a given allocation belongs to such sets?

\emph{Submodularity} (or \emph{convexity}) is a crucial property which yields interesting answers to some of the questions above. An instance of a cooperative cost game is called submodular if the characteristic function $\nu$ is submodular, meaning that 
\[\forall A, B \subseteq N, \; \nu(A) + \nu(B) \geq \nu (A \cup B) + \nu(A \cap B). \qquad (*)\]
Submodular games exhibit a large number of desirable properties. In particular,
(i) a core solution always exists and can be computed in polynomial time \cite{journal/ijgt/Shapley71}; 
(ii) testing whether an allocation belongs to the core is equivalent to separating over the extended polymatroid of $\nu$, which can be performed efficiently \cite{book/algorithms/GroetschelLS88}; 
(iii) computing the nucleolus can be done efficiently \cite{report/Kuipers96}; 
(iv) there is a nice ``snowballing'' effect that arises when the game is played cooperatively, meaning that joining a coalition becomes more attractive as the coalition grows, and so the value of the so-called grand coalition $\nu(N)$ is always reached \cite{journal/ijgt/Shapley71}. We refer to \cite{journal/ijgt/MaschlerPS71,journal/ijgt/Shapley71} for other interesting properties of submodular games involving other crucial solution concepts. Given these observations, it is not surprising that some researchers have investigated whether it is possible to give an efficient characterization of submodular instances, for prominent cooperative games that are in general non-convex. Such characterizations are known, for example, for the minimum coloring game and the minimum vertex cover game \cite{journal/mmor/Okamoto03}, as well as for some communication games \cite{journal/ijgt/vandenNouwelandB91}.

This paper focuses on one of the most fundamental cooperative games, namely the \emph{spanning tree game}. This game was introduced more than 40 years ago \cite{journal/networks/Bird76,journal/networks/ClausK73}, and since then it has been widely studied in the literature. To get an intuition about the problem, consider the following setting. A set of clients $N$ would like to be connected to a central source $r$ which can provide a service to them. The clients wish to build a network connecting them to the source $r$, at minimum cost. An obvious way to solve this problem is to compute a minimum spanning tree connecting $N\cup\set{r}$. But how should the clients fairly split the cost of the tree among them? Formally, an instance of the spanning tree game is described by an edge-weighted complete graph $G=(V,E)$ where $V=N\cup\set{r}$. The set of players is given by $N$, and the characteristic function $\nu(S)$ is equal to the cost of a minimum spanning tree in the subgraph induced by $S \cup \{r\}$.

Despite being one of the most studied cooperative games, the existence of an efficient characterization of submodularity for the spanning tree game has remained elusive so far. Granot and Huberman \cite{journal/siamadm/GranotH82} proved that spanning tree games are \emph{permutationally convex} (which is a generalization of submodularity). Their result implies that a core solution always exists for such games, despite being non-convex in general (this was first proven by the same authors in \cite{journal/mp/GranotH81}). However, other nice properties of submodular games do not generalize: for general spanning tree games, testing core membership is coNP-hard \cite{journal/ijgt/FaigleKFH97}, and computing the nucleolus is NP-hard \cite{journal/ijgt/FaigleKK98}. Trudeau \cite{journal/geb/Trudeau12} gave a sufficient condition for an instance of the game to be submodular. An important step forward was made by Kobayashi and Okamoto \cite{journal/networks/KobayashiO14}, who gave a characterization of submodularity for instances of the spanning tree game where the edge weights are restricted to take only two values. For general weights, they stated some necessary (but not always sufficient) as well as some sufficient (but not always necessary) conditions for an instance to be submodular. Whether a polynomial time characterization of submodularity exists for spanning tree games is left as an open question. In fact, they stated twice in their paper:

\smallskip
\noindent
\emph{``We feel that recognizing a submodular minimum-cost spanning tree game is coNP-complete, but we are still far from proving such a result.''}\\  

\medskip
\noindent
{\bf Our results and techniques.} In this paper, we finally settle this open question: we give a polynomial time characterization
 of submodular spanning tree games. 
 
Our characterization uses combinatorial techniques and it is based on two main ingredients. The first one, described in Section \ref{sec:violated_cycles}, is a generalization of Kobayashi and Okamoto's result \cite{journal/networks/KobayashiO14}. When the edges can have only two distinct weights, they proved that the only obstruction to submodularity comes from the presence of certain cycles in the graph induced by the cheaper edges. 
When dealing with more weight values, say $w_1<w_2<\dots<w_k$, things become necessarily more complicated.  We can still prove that an obstruction to submodularity is given by certain cycles, which we call \emph{violated}, but (a) our definition of violated cycles is more involved than the one in \cite{journal/networks/KobayashiO14}, and (b) we have to look for such cycles not just in one induced graph, but in each graph induced by the edges of weight at most $w_i$, for all $i<k$. 

Furthermore, the presence of violated cycles is not anymore the only obstruction to submodularity. Roughly speaking, violated cycles capture how the edges of a certain weight should relate to the cheaper ones, but we still need a condition that takes into account the ``magnitude'' of distinct weight values, when $k>2$. This leads to the second main ingredient of our characterization, described in Section \ref{sec:candidate_edges}. We show that, under the assumption of not having violated cycles, we can identify polynomially many subsets of vertices which could yield the highest possible violation to the submodularity inequality $(*)$. We can then efficiently test the submodularity of our instance by checking whether the inequality $(*)$ is satisfied on this family of subsets of vertices.
Combining these two ingredients yields a polynomial time characterization of submodularity for spanning tree games, as described in Section \ref{sec:characterization}.

We conclude our paper with an additional result. As previously mentioned,
the authors of \cite{journal/networks/KobayashiO14} gave a necessary condition for submodularity of the spanning tree game. They also stated that they do not know whether their condition can be verified in polynomial time. We answer this question affirmatively in Section~\ref{sec:wide}.

\section{Preliminaries and notation}
For a subset $S\subseteq V$, let $\textsf{mst}(S)$ denote the weight of a minimum spanning tree in $G[S]$, where $G[S]$ is the subgraph of $G$ induced by $S$. Given a subgraph $H$ of $G$, let $w(H)$ denote the sum of edge weights in $H$, i.e.~$\sum_{e\in E(H)}w(e)$. For an edge set $F$, we will also use $w(F)$ to indicate the sum of edge weights in $F$. For a vertex $u\in V$, $N_H(u)$ is the neighborhood of $u$ in $H$, while $\delta_H(u)$ is the set of edges incident to $u$ in $H$. For any positive integer $k$, $[k]$ represents the set $\set{1,2,\dots,k}$. Given a pair of vertices $u,v\in N$, let $\mathcal{S}_{uv}$ denote the family of vertex subsets which contain $r$ but not $u$ or $v$, i.e.~$\mathcal{S}_{uv}:=\set{S\subseteq V:r\in S \mbox{ and } u,v\notin S}$. Define the function $f_{uv}:\mathcal{S}_{uv}\rightarrow \R$ as
\[f_{uv}(S) := \textsf{mst}(S\cup u) + \textsf{mst}(S\cup v) - \textsf{mst}(S) - \textsf{mst}(S\cup\set{u,v}).\]
It is easy to see that the spanning tree game on $G$ is submodular if and only if $f_{uv}(S)\geq 0$ for all $u,v\in N$ and $S\in \mathcal{S}_{uv}$. Let $w_1<w_2<\dots<w_k$ be the edge-weights of $G$. For each $i\in [k]$, define the graph $G_i:=(V,E_i)$ where $E_i:=\set{e\in E:w(e)\leq w_i}$. Note that $G_k=G$.
For a vertex $u\in V$, denote $N_i(u)$ as the neighborhood of $u$ in $G_i$. For an edge $uv\in E$, define the neighborhood of $uv$ in $G_i$ as 
\[N_i(uv) := N_i(u) \cap N_i(v).\]
It represents the set of vertices whose edges to $u$ and $v$ have weight at most $w_i$. Notice that $u,v\notin N_i(uv)$. We will also need the following graph theory terminology. A \emph{hole} is an induced cycle with at least four vertices. A \emph{diamond} is the complete graph $K_4$ minus one edge. We will refer to the vertices of degree 2 in a diamond as \emph{tips}. Lastly, the following property of minimum spanning trees will be useful to us.

\begin{lemma}\label{lem:reuse}
Let $T$ be a minimum spanning tree of $G$. For every subset $S\subseteq V$, there exists a minimum spanning tree of $G[S]$ which contains $E(T[S])$.
\end{lemma}

\begin{proof}
Let $T_S$ be a minimum spanning tree of $G[S]$. We proceed by induction on $\size{E(T[S])\setminus E(T_S)}$. For the base case, if $\size{E(T[S])\setminus E(T_S)}=0$, then $T_S$ contains all the edges in $E(T[S])$. For the inductive step, assume $\size{E(T[S])\setminus E(T_S)}>0$. Then, there exists an edge $e\in E(T[S])$ where $e\notin E(T_S)$. Adding $e$ to $T_S$ creates a cycle $C$ in $G[S]$. So there exists an edge $f\in E(C)$ which is not an edge of $T$. Pick an appropriate $f$ such that when added to $T$ creates a cycle containing $e$. Since $T+f-e$ is a spanning tree of $G$, we have $w(f)\geq w(e)$. On the other hand, since $T_S+e-f$ is a spanning tree of $G[S]$, we also have $w(e)\geq w(f)$. This implies $w(e)=w(f)$, so $T_S+e-f$ is a minimum spanning tree of $G[S]$. As $\size{E(T[S])\setminus E(T_S+e-f)} = \size{E(T[S])\setminus E(T_S)}-1$, by the induction hypothesis we are done.
\end{proof}

\section{Violated cycles}
\label{sec:violated_cycles}

In this section, we will prove that a submodular spanning tree game does not contain violated cycles, which will be defined later. First, we need to introduce the concept of \emph{well-covered} cycles.

\begin{definition}
Given a cycle $C$ and a chord $f=uv$, let $P_1$ and $P_2$ denote the two $u$-$v$ paths in $C$. The cycles $P_1+f$ and $P_2+f$ are called the \emph{subcycles of $C$ formed by $f$}. We say that $f$ \emph{covers} $C$ if $w(f)\geq w(e)$ for all $e\in E(P_1)$ or for all $e\in E(P_2)$. If $C$ is covered by all of its chords, then it is \emph{well-covered}.
\end{definition}

Next, we define the following two simple structures. We then proceed to show that a submodular spanning tree game does not contain either of them.

\begin{definition}
A hole is \emph{bad} if at least one of its vertices is not adjacent to $r$. An induced diamond is \emph{bad} if its hamiltonian cycle is well-covered but at least one of its tips is not adjacent to $r$. 
\end{definition}

\begin{figure}[H]
\def\dist{1.5}
\centering
\begin{tikzpicture}[node distance=\dist cm, inner sep=2.5pt, minimum size=2.5pt, auto]
\node [node] (u1) at (0,0) {};
\node [node,label=right:\footnotesize{$r$}] (u2) [right of=u1] {};
\node [node] (u3) [below of=u2] {};
\node [node] (u4) [left of=u3] {};
\path [edge] (u1) -- (u2) -- (u3) -- (u4) -- (u1);

\node [node] (v1) at (2*\dist,0) {};
\node [node] (v2) [right of=v1] {};
\node [node] (v3) [below of=v2] {};
\node [node] (v4) [left of=v3] {};
\node [node,label=right:\footnotesize{$r$}] (v5) [right of=v2] {};
\path [edge] (v1) -- (v2) -- (v3) -- (v4) -- (v1);
\path [edge] (v5) -- (v2);
\path [edge] (v5) -- (v3);
\path [edge] (v5) -- (v4);

\node [node] (x1) at (5*\dist,0) {};
\node [selected vertex,label=right:\footnotesize{$r$}] (x2) [right of=x1] {};
\node [node] (x3) [below of=x2] {};
\node [selected vertex] (x4) [left of=x3] {};
\path [edge] (x1) -- (x2) -- (x3) -- (x4) -- (x1);
\path [edge] (x1) -- (x3);

\node [node] (w1) at (7*\dist,0) {};
\node [selected vertex] (w2) [right of=w1] {};
\node [node] (w3) [below of=w2] {};
\node [selected vertex] (w4) [left of=w3] {};
\node [node,label=right:\footnotesize{$r$}] (w5) [right of=w2] {};
\path [edge] (w1) -- (w2) -- (w3) -- (w4) -- (w1);
\path [edge] (w1) -- (w3);
\path [edge] (w5) -- (w2);
\path [edge] (w5) -- (w3);
\end{tikzpicture}
\caption{Examples of bad holes and bad induced diamonds. The tips of the diamonds are shaded. Every edge here has the same weight.}
\end{figure}
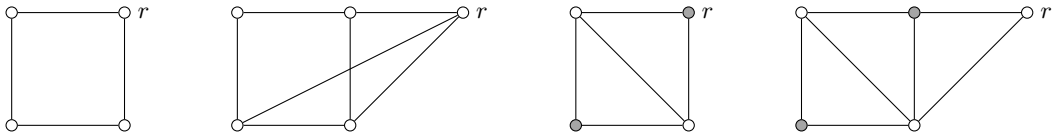

\begin{lemma}\label{lem:holediamond}
If the spanning tree game on $G$ is submodular, then there are no (a) bad holes or (b) bad induced diamonds in $G_i$ for all $i<k$.
\end{lemma}

\begin{proof}[Proof of Lemma \ref{lem:holediamond}(a)]
We will prove the contrapositive. Let $C$ be a bad hole in $G_i$ for some $i<k$. Consider the following two cases:

\medskip
\noindent
\emph{Case 1: $C$ contains $r$.} Let $u,v$ be the vertices adjacent to $r$ in $C$. Define the set $S:=V(C)\setminus\set{u,v}$. To prove that the instance is not submodular, it suffices to show that $f_{uv}(S)<0$. Let $P$ be the path obtained by deleting $r,u,v$ from $C$. Let $u',v'$ be the endpoints of $P$ where $uu',vv'\in E(C)$ (see Figure \ref{fig:holes}). Note that $u'=v'$ if $P$ is a singleton. It is easy to see that
\begin{align*}
	\textsf{mst}(S\cup u) &= w(P) + w(ru) + w(uu') \\
	\textsf{mst}(S\cup v) &= w(P) + w(rv) + w(vv').
\end{align*}
By Lemma \ref{lem:reuse}, there exists a minimum spanning tree of $G[S]$ which contains $P$. Thus,
\[\textsf{mst}(S) \geq w(P) + w_{i+1}.\]
Next, deleting the most expensive edge from $C$ creates a minimum spanning tree of $G[S\cup\set{u,v}]$. Since this edge has weight at most $w_i$, we obtain
\[\textsf{mst}(S\cup\set{u,v}) \geq w(P) + w(ru) + w(uu') + w(rv) + w(vv') - w_i\]
Then, combining the equations and inequalities above yields
\[f_{uv}(S) = \textsf{mst}(S\cup u) + \textsf{mst}(S\cup v) - \textsf{mst}(S) - \textsf{mst}(S\cup\set{u,v}) \leq w_i - w_{i+1} < 0.\]

\begin{figure}[ht]
\def \hrad {1.5}
\def \vrad {0.8}
\centering
\begin{minipage}{0.33\textwidth}
\centering
\begin{tikzpicture}
	\fill[gray!30] plot[smooth, tension=.7] coordinates {(0,-1.2*\vrad) (-\hrad,-0.8*\vrad) (-1.4*\hrad,0.2) (-0.3*\hrad,0.2*\vrad) (-0.3*\hrad,1.2*\vrad) (0.3*\hrad,1.2*\vrad) (0.3*\hrad,0.2*\vrad) (1.4*\hrad,0.2) (\hrad,-0.8*\vrad) (0,-1.2*\vrad) };

	\draw[thick,color=black!50] (0,0) ellipse (\hrad cm and \vrad cm);

	\node[node,label=below:\footnotesize{$r$}] (r) at (0,\vrad) {};
	\node[node,label=below:\footnotesize{$u$}] (u) at (-\hrad/1.7,0.8*\vrad) {};
	\node[node,label=left:\footnotesize{$u'$}] (u') at (-\hrad,0) {};
	\node[node,label=below:\footnotesize{$v$}] (v) at (\hrad/1.7,0.8*\vrad) {};
	\node[node,label=right:\footnotesize{$v'$}] (v') at (\hrad,0) {};
	\node at (0,2*\vrad) {}; 

	\node at (0,-1.4*\vrad) {\small{$P$}};
\end{tikzpicture}
\end{minipage}
\begin{minipage}{0.33\textwidth}
\centering
\begin{tikzpicture}
	\fill[gray!30] plot[smooth, tension=.7] coordinates {(0,-1.2*\vrad) (-\hrad,-0.8*\vrad) (-1.4*\hrad,0.2) (-0.4*\hrad,0.2*\vrad) (-0.4*\hrad,2.2*\vrad) (0.4*\hrad,2.2*\vrad) (0.4*\hrad,0.2*\vrad) (1.4*\hrad,0.2) (\hrad,-0.8*\vrad) (0,-1.2*\vrad) };

	\draw[thick,color=black!50] (0,0) ellipse (\hrad cm and \vrad cm);

	\node[node,label=left:\footnotesize{$r$}] (r) at (0,2*\vrad) {};
	\node[node,label=below:\footnotesize{$s$}] (s) at (0,\vrad) {};
	\node[node,label=below:\footnotesize{$u$}] (u) at (-\hrad/1.7,0.8*\vrad) {};
	\node[node,label=left:\footnotesize{$u'$}] (u') at (-\hrad,0) {};
	\node[node,label=below:\footnotesize{$v$}] (v) at (\hrad/1.7,0.8*\vrad) {};
	\node[node,label=right:\footnotesize{$v'$}] (v') at (\hrad,0) {};

	\path[edge] (r) -- (s);

	\node at (0,-1.4*\vrad) {\small{$P$}};
\end{tikzpicture}
\end{minipage}
\begin{minipage}{0.32\textwidth}
\centering
\begin{tikzpicture}
	\fill[gray!30] plot[smooth, tension=.7] coordinates {(0,-1.2*\vrad) (-\hrad,-0.8*\vrad) (-1.4*\hrad,0.2) (-0.4*\hrad,0.2*\vrad) (-0.4*\hrad,2.2*\vrad) (0.4*\hrad,2.2*\vrad) (0.4*\hrad,0.2*\vrad) (1.4*\hrad,0.2) (\hrad,-0.8*\vrad) (0,-1.2*\vrad) };

	\draw[thick,color=black!50] (0,0) ellipse (\hrad cm and \vrad cm);

	\node[node,label=left:\footnotesize{$r$}] (r) at (0,2*\vrad) {};
	\node[node,label=below:\footnotesize{$s$}] (s) at (0,\vrad) {};
	\node[node,label=below:\footnotesize{$u$}] (u) at (-\hrad/1.7,0.8*\vrad) {};
	\node[node,label=left:\footnotesize{$u'$}] (u') at (-\hrad,0) {};
	\node[node,label=below:\footnotesize{$v$}] (v) at (\hrad/1.7,0.8*\vrad) {};
	\node[node,label=right:\footnotesize{$v'$}] (v') at (\hrad,0) {};

	\path[edge] (r) -- (s);
	\path[edge] (r) -- (u);

	\node at (0,-1.4*\vrad) {\small{$P$}};
\end{tikzpicture}
\end{minipage}
\caption{An example of the subgraph $G_i[S\cup\set{u,v}]$ in Case 1 and Subcases 2.1--2.2 respectively. The shaded region represents the set $S$.}
\label{fig:holes}
\end{figure}
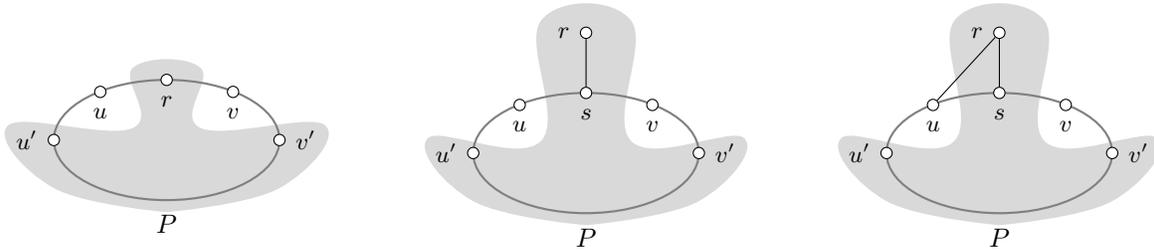

\noindent
\emph{Case 2: $C$ does not contain $r$.} We claim that if $r$ is adjacent in $G_i$ to two non-adjacent vertices of $C$, then we are done. Since $C$ is a bad hole in $G_i$, there exists a vertex $t\in V(C)$ such that $rt\notin E_i$. Starting from $t$, traverse the hole $C$ in both directions until we encounter the first vertices $p$ and $q$ such that $rp,rq\in E_i$ respectively. By our assumption, $p\neq q$ and $pq\notin E(C)$. Let $Q$ be the $p$-$q$ path in $C$ which contains $t$. Then, by our choice of $p$ and $q$, the cycle $Q+rp+rq$ is a bad hole which contains $r$. It follows that the instance is not submodular by Case 1.

Pick a vertex $s\in V(C)$ with the cheapest edge to $r$, i.e. $w(rs)\leq w(rx)$ for all $x\in V(C)$. Let $u,v$ be the vertices adjacent to $s$ in $C$. Define the set $S:=(V(C)\cup r)\setminus\set{u,v}$. Similar to the previous case, it suffices to show that $f_{uv}(S)<0$. Let $P$ be the path obtained by deleting $s,u,v$ from $C$. Let $u',v'$ be the endpoints of $P$ where $uu',vv'\in E(C)$ (see Figure \ref{fig:holes}). We are left with the following two subcases:

\medskip
\noindent
\emph{Subcase 2.1: $r$ is adjacent to at most one vertex of $C$.} It is easy to see that
\begin{align*}
	\textsf{mst}(S\cup u) &= w(P) + w(rs) + w(su) + w(uu') \\
	\textsf{mst}(S\cup v) &= w(P) + w(rs) + w(sv) + w(vv').
\end{align*}
By Lemma \ref{lem:reuse}, there exists a minimum spanning tree of $G[S]$ which contains $P$ and $rs$. Thus, 
\[\textsf{mst}(S) \geq w(P) + w(rs) + w_{i+1}\]
Next, deleting the most expensive edge from $C$ and adding $rs$ creates a minimum spanning tree of $G[S\cup\set{u,v}]$. Since the deleted edge has weight at most $w_i$, we obtain
\[\textsf{mst}(S\cup\set{u,v}) \geq w(P) + w(rs) + w(su) + w(uu') + w(sv) + w(vv') - w_i.\]
Then, combining the equations and inequalities above yields
\[f_{uv}(S) = \textsf{mst}(S\cup u) + \textsf{mst}(S\cup v) - \textsf{mst}(S) - \textsf{mst}(S\cup\set{u,v}) \leq w_i - w_{i+1} < 0.\]

\smallskip
\noindent
\emph{Subcase 2.2: $r$ is adjacent to two adjacent vertices of $C$.} Without loss of generality, suppose that $rs,ru\in E_i$. It is easy to see that
\begin{align*}
	\textsf{mst}(S\cup u) &= w(P) + w(rs) + \min\set{w(ru),w(su)} + w(uu') \\
	\textsf{mst}(S\cup v) &= w(P) + w(rs) + w(sv) + w(vv')
\end{align*}
By Lemma \ref{lem:reuse}, there exists a minimum spanning tree of $G[S]$ which contains $P$ and $rs$. Thus,
\[\textsf{mst}(S) \geq w(P) + w(rs) + w_{i+1}.\]
Next, deleting the most expensive edge from the triangle $\set{rs,ru,su}$ does not increase the value of a minimum spanning tree in $G[S\cup\set{u,v}]$. In fact, a minimum spanning tree can be obtained by deleting one more edge from $G_i[S\cup\set{u,v}]$. So,
\[\textsf{mst}(S\cup\set{u,v}) \geq w(P) + w(rs) + \min\set{w(ru),w(su)} + w(uu') + w(sv) + w(vv') - w_i\]
Then, combining the equations and inequalities above yields
\[f_{uv}(S) = \textsf{mst}(S\cup u) + \textsf{mst}(S\cup v) - \textsf{mst}(S) - \textsf{mst}(S\cup\set{u,v}) \leq w_i - w_{i+1} < 0.\]
\end{proof}

\begin{proof}[Proof of Lemma \ref{lem:holediamond}(b)]
We will again prove the contrapositive. Let $D$ be a bad induced diamond in $G_i$ for some $i<k$. Consider the following two cases:

\medskip
\noindent
\emph{Case 1: $D$ contains $r$.} Observe that $r$ is a tip of $D$. Let $s$ be the other tip and $u,v$ be the non-tip vertices of $D$. Define the set $S:=\set{r,s}$ (see Figure \ref{fig:diamonds}). To prove that the instance is not submodular, it suffices to show that $f_{uv}(S)<0$. It is easy to see that
\begin{align*}
	\textsf{mst}(S) &\geq w_{i+1} \\
	\textsf{mst}(S\cup u) &= w(ru) + w(su) \\
	\textsf{mst}(S\cup v) &= w(rv) + w(sv).	
\end{align*}
Since the hamiltonian cycle of $D$ is well-covered, its chord $uv$ can be deleted without increasing the value of a minimum spanning tree in $G[S\cup\set{u,v}]$. We are now left with the hamiltonian cycle of $D$, so a minimum spanning tree can be obtained by removing the most expensive edge. This gives  
\[\textsf{mst}(S\cup\set{u,v}) \geq w(ru) + w(su) + w(rv) + w(sv) - w_i.\]
Then, combining the equations and inequalities above yields
\[f_{uv}(S) = \textsf{mst}(S\cup u) + \textsf{mst}(S\cup v) - \textsf{mst}(S) - \textsf{mst}(S\cup\set{u,v}) \leq w_i - w_{i+1} < 0.\]

\begin{figure}[ht]
\centering
\def\dist{1.2}
\begin{minipage}{0.24\textwidth}
\centering
\begin{tikzpicture}[node distance=\dist cm]
	\fill[gray!30] (0,-\dist/1.5) ellipse (0.3*\dist cm and 1.4*\dist cm);

	\node[node,label=above:\footnotesize{$r$}] (r) {};
	\node[node,label=left:\footnotesize{$u$}] (u) [below left of=r] {};
	\node[node,label=right:\footnotesize{$v$}] (v) [below right of=r] {};
	\node[node,label=below:\footnotesize{$s$}] (s) [below right of=u] {};
	\node [above of=r, yshift=0.3*\dist cm] {}; 

	\path[edge] (r) -- (u) -- (s) -- (v) -- (r);
	\path[edge] (u) -- (v);
\end{tikzpicture}
\end{minipage}
\begin{minipage}{0.24\textwidth}
\centering
\begin{tikzpicture}[node distance=\dist cm]
	\fill[gray!30] (0,-\dist/3) ellipse (0.3*\dist cm and 1.7*\dist cm);

	\node[node,label=left:\footnotesize{$s$}] (s) {};
	\node[node,label=left:\footnotesize{$r$}] (r) [above of=s] {};
	\node[node,label=left:\footnotesize{$u$}] (u) [below left of=s] {};
	\node[node,label=right:\footnotesize{$v$}] (v) [below right of=s] {};
	\node[node,label=below:\footnotesize{$t$}] (t) [below right of=u] {};

	\path[edge] (s) -- (u) -- (t) -- (v) -- (s);
	\path[edge] (u) -- (v);
	\path[edge] (r) -- (s);
\end{tikzpicture}
\end{minipage}
\begin{minipage}{0.24\textwidth}
\centering
\begin{tikzpicture}[node distance=\dist cm]
	\fill[gray!30] (0,-\dist/3) ellipse (0.3*\dist cm and 1.7*\dist cm);

	\node[node,label=left:\footnotesize{$s$}] (s) {};
	\node[node,label=left:\footnotesize{$r$}] (r) [above of=s] {};
	\node[node,label=left:\footnotesize{$u$}] (u) [below left of=s] {};
	\node[node,label=right:\footnotesize{$v$}] (v) [below right of=s] {};
	\node[node,label=below:\footnotesize{$t$}] (t) [below right of=u] {};

	\path[edge] (s) -- (u) -- (t) -- (v) -- (s);
	\path[edge] (u) -- (v);
	\path[edge] (r) -- (s);
	\path[edge] (r) -- (u);
\end{tikzpicture}
\end{minipage}
\begin{minipage}{0.24\textwidth}
\centering
\begin{tikzpicture}[node distance=\dist cm]
	\fill[gray!30] (0,-\dist/3) ellipse (0.3*\dist cm and 1.7*\dist cm);

	\node[node,label=left:\footnotesize{$s$}] (s) {};
	\node[node,label=left:\footnotesize{$r$}] (r) [above of=s] {};
	\node[node,label=left:\footnotesize{$u$}] (u) [below left of=s] {};
	\node[node,label=right:\footnotesize{$v$}] (v) [below right of=s] {};
	\node[node,label=below:\footnotesize{$t$}] (t) [below right of=u] {};

	\path[edge] (s) -- (u) -- (t) -- (v) -- (s);
	\path[edge] (u) -- (v);
	\path[edge] (r) -- (s);
	\path[edge] (r) -- (u);
	\path[edge] (r) -- (v);
\end{tikzpicture}
\end{minipage}
\caption{An example of the subgraph $G_i[S\cup\set{u,v}]$ in Case 1 and Subcases 2.1--2.3 respectively. The shaded region represents the set $S$. For Subcase 2.1, the picture assumes that $w(rs)< w(ru)$.}
\label{fig:diamonds}
\end{figure}
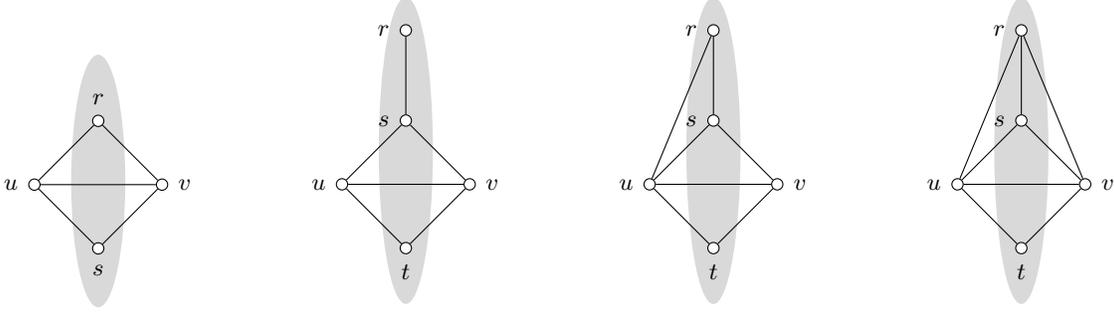

\noindent
\emph{Case 2: $D$ does not contain $r$.} Let $s,t$ be the tips of $D$ such that $w(rs)\leq w(rt)$. Note that $rt\notin E_i$ because $D$ is a bad induced diamond in $G_i$. Let $u,v$ be the non-tip vertices of $D$ where $w(ru)\leq w(rv)$. Define the set $S:=\set{r,s,t}$ (see Figure \ref{fig:diamonds}). Similar to the previous case, it suffices to show that $f_{uv}(S)<0$. Consider the following three subcases:

\medskip
\noindent
\emph{Subcase 2.1: $r$ is adjacent to at most one vertex of $D$.} Note that $rv\notin E_i$. It is also easy to see that
\begin{align*}
	\textsf{mst}(S) &\geq w(rs) + w_{i+1} \\
	\textsf{mst}(S\cup u) &= \min\set{w(rs),w(ru)} + w(su) + w(tu) \\
  	\textsf{mst}(S\cup v) &= \min\set{w(rs),w(rv)} + w(sv) + w(tv) 
\end{align*}
Next, observe that we can delete $uv$ and the most expensive edge in the hamiltonian cycle of $D$ without increasing the value of a minimum spanning tree in $G[S\cup\set{u,v}]$. Therefore,
\[\textsf{mst}(S\cup\set{u,v}) \geq \min\set{w(rs),w(ru)} + w(su) + w(tu) + w(sv) + w(tv) - w_i.\]
Then, combining the equations and inequalities above yields
\[f_{uv}(S) = \textsf{mst}(S\cup u) + \textsf{mst}(S\cup v) - \textsf{mst}(S) - \textsf{mst}(S\cup\set{u,v}) \leq w_i - w_{i+1} < 0.\]

\smallskip
\noindent
\emph{Subcase 2.2: $r$ is adjacent to two vertices of $D$.} We claim that if $ru,rv\in E_i$, then we are done. Note that this implies $rs,rt\notin E_i$. So if $w(uv)\geq \max\set{w(su),w(sv)}$, then $G[\set{r,s,u,v}]$ is a bad induced diamond in $G_i$. Otherwise, $G[\set{r,t,u,v}]$ is a bad induced diamond in $G_i$. Since they both contain $r$, the instance is not submodular by Case 1. Thus, we may assume that $rs,ru\in E_i$. Additionally, we may assume that $w(su)<\max\set{w(rs),w(ru)}$. Otherwise, $G[\set{r,s,u,v}]$ is a bad induced diamond, and we are done again by Case 1. Then, it is easy to see that
\begin{align*}
  \textsf{mst}(S) &\geq w(rs) + w_{i+1} \\
  \textsf{mst}(S\cup u) &= \min\set{w(rs),w(ru)} + w(su) + w(tu) \\
  \textsf{mst}(S\cup v) &= w(rs) + w(sv) + w(tv)
\end{align*}
By a similar reasoning as before, we can delete $uv$ and the most expensive edge in the hamiltonian cycle of $D$ without increasing the value of a minimum spanning tree in $G[S\cup\set{u,v}]$. Hence,
\begin{align*}
  \textsf{mst}(S\cup\set{u,v}) &\geq \min\set{w(rs),w(ru)} + w(su) + w(tu) + w(sv) + w(tv) - w_i
\end{align*}
Finally, combining the equations and inequalities above yields
\[f_{uv}(S) = \textsf{mst}(S\cup u) + \textsf{mst}(S\cup v) - \textsf{mst}(S) - \textsf{mst}(S\cup\set{u,v}) \leq w_i - w_{i+1} < 0.\]

\smallskip
\noindent
\emph{Subcase 2.3: $r$ is adjacent to three vertices of $D$.} Let $w(rv)=w_j$ for some $j\leq i$, and consider the induced diamond $G[\set{r,t,u,v}]$. If it is well-covered, then we are done by Case 1 because $rt\notin E_i$. So we may assume that $\max\set{w(su),w(sv)} \leq w(uv) < w(rv)$. Additionally, we may assume that $w(su)<\max\set{w(rs),w(ru)}$. Otherwise, $G[\set{r,s,u,v}]$ is a bad induced diamond in $G_{j-1}$, and we are done again by Case 1. Then, it is easy to see that  
\begin{align*}
  \textsf{mst}(S) &\geq w(rs) + w_{i+1} \\
  \textsf{mst}(S\cup u) &= \min\set{w(rs),w(ru)} + w(su) + w(tu) \\
  \textsf{mst}(S\cup v) &= \min\set{w(rs),w(rv)} + w(sv) + w(tv)
\end{align*}
By a similar reasoning as before, we can delete $uv$ and the most expensive edge in the hamiltonian cycle of $D$ without increasing the value of a minimum spanning tree in $G[S\cup\set{u,v}]$. Hence, 
\[\textsf{mst}(S\cup\set{u,v}) \geq \min\set{w(rs),w(ru)} + w(su) + w(tu) + w(sv) + w(tv) - w_i\]
Finally, combining the equations and inequalities above yields 
\[f_{uv}(S) = \textsf{mst}(S\cup u) + \textsf{mst}(S\cup v) - \textsf{mst}(S) - \textsf{mst}(S\cup\set{u,v}) \leq w_i - w_{i+1} < 0.\]
\end{proof}

We are now ready to define the main object of study in this section:

\begin{definition}
A \emph{violated} cycle is a well-covered cycle which contains at least a pair of non-adjacent vertices and at least a vertex not adjacent to $r$.
\end{definition}

Observe that bad holes and hamiltonian cycles of bad induced diamonds are examples of violated cycles (we consider a hole to be well-covered). The next lemma extends the scope of Lemma \ref{lem:holediamond} to include violated cycles. When $k=2$, this coincides with the condition given by Kobayashi and Okamoto \cite{journal/networks/KobayashiO14} because every cycle in $G_1$ is well-covered.

\begin{lemma}\label{lem:violated}
If the spanning tree game on $G$ is submodular, then there are no violated cycles in $G_i$ for all $i<k$.
\end{lemma}

\begin{proof}
We will prove the contrapositive. Let $j$ be the smallest integer such that $G_j$ contains a violated cycle. By our choice of $j$, there are no violated cycles in $G_i$ for all $i<j$. Let $C$ be a smallest violated cycle in $G_j$. Then, $\max_{e\in E(C)}w(e)=w_j$. We first prove the following claim:

\begin{claim}\label{clm:subcycles}
For any chord $f$, the subcycles of $C$ formed by $f$ are well-covered.
\end{claim}

\begin{proof}
Let $C_1$ and $C_2$ denote the subcycles of $C$ formed by $f$. For the purpose of contradiction, suppose $C_2$ is not well-covered. Let $g=uv$ be the cheapest chord in $C_2$ such that $w(g)<w(f)$ and $w(g)<w(h)$ for some edge $h\in E(C_2)$, where $f$ and $h$ lie in different subcycles of $C_2$ formed by $g$ (see Figure \ref{fig:subcycles} for an example). This chord exists because $C$ is well-covered but $C_2$ is not. Consider the subcycles $C_3$ and $C_4$ of $C$ formed by $g$, where $f$ is a chord of the former. Observe that $C_3$ is well-covered because $w(g)<w(h)$, while $C_4$ is well-covered due to our choice of $g$. Moreover, we have $w(g)\geq w(e)$ for all $e\in E(C_3)$ as $C$ is well-covered. Let $w(g)=w_\ell$ for some $\ell<j$. Then, $C_3$ is still present in $G_\ell$ but not $f$ because $w(g)<w(f)$. Thus, the vertices of $C_3$ are adjacent to $r$ in $G_\ell$ because there are no violated cycles in $G_\ell$. In particular, we have $ru,rv\in E_\ell$. Next, since $C$ is a violated cycle in $G_j$, there exists a vertex $s\in V(C_4)\setminus V(C_3)$ such that $rs\notin E_j$. This implies that the vertices of $C_4$ are pairwise adjacent in $G_j$, as otherwise it is a smaller violated cycle than $C$. In particular, we have $su,sv\in E_j$. Now, consider the 4-cycle $D$ defined by $E(D):=\set{ru,rv,su,sv}$. It is well-covered because $w(g)=w_\ell$ and $ru,rv\in E_\ell$. As $rs\notin E_j$, it is a violated cycle in $G_j$. However, it is smaller than $C$ because $C_3$ has at least 4 vertices. We have arrived at a contradiction.	
\end{proof}

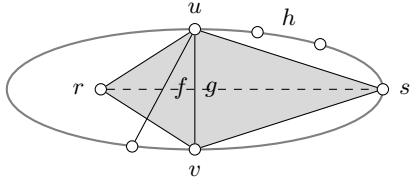
\begin{figure}
\centering
\def \hrad {2.5}
\def \vrad {0.8}
\begin{tikzpicture}
	\path [fill=gray!30] (-\hrad/2,0) -- (0,\vrad) -- (\hrad,0) -- (0,-\vrad) -- (-\hrad/2,0);

	\draw[thick,color=black!50] (0,0) ellipse (\hrad cm and \vrad cm);

	\node[node,label=left:\footnotesize{$r$}] (r) at (-\hrad/2,0) {};
	\node[node,label=right:\footnotesize{$s$}] (s) at (\hrad,0) {};
	\node[node,label=above:\footnotesize{$u$}] (u) at (0,\vrad) {};
	\node[node,label=below:\footnotesize{$v$}] (v) at (0,-\vrad) {};
	\node[node] (v1) at (\hrad/3,0.95*\vrad) {};
	\node[node] (v2) at (\hrad/1.5,0.75*\vrad) {};
	\node[node] (v3) at (-\hrad/3,-0.95*\vrad) {};

	\node [font=\footnotesize] (h) at (0.5*\hrad,1.2*\vrad) {$h$};

	\path[dashed edge] (r) -- (s);
	\path[edge] (r) -- (u);
	\path[edge] (r) -- (v);
	\path[edge] (s) -- (u);
	\path[edge] (s) -- (v);
	\path[edge] (u) -- node [weight,right] {\footnotesize{$f$}} (v3);
	\path[edge] (u) -- node [weight,right] {\footnotesize{$g$}} (v);

\end{tikzpicture}
\caption{The ellipse represents the violated cycle $C$ in Claim \ref{clm:subcycles}. The shaded region highlights the smaller violated cycle $D$. The dashed edge indicates $rs\notin E_j$.}
\label{fig:subcycles}
\end{figure}

Our goal is to show the existence of a bad hole or a bad induced diamond in $G_j$. Then, we can invoke Lemma \ref{lem:holediamond} to conclude that the game is not submodular. We may assume that $C$ has a chord, otherwise it is trivially a bad hole. First, consider the case when $r\in V(C)$. Let $s\in V(C)$ where $rs\notin E_j$. For any chord $f$ in $C$, observe that $r$ and $s$ lie in different subcycles of $C$ formed by $f$. This is because the subcycles are well-covered by the previous claim, so the one which contains both $r$ and $s$ will contradict the minimality of $C$. Now, let $g$ be a chord of $C$. Let $C_r$ and $C_s$ denote the subcycles of $C$ formed by $g$ where $r\in V(C_r)$ and $s\in V(C_s)$. Observe that the vertices of $C_r$ are adjacent to $r$ due to the minimality of $C$. Thus, $C_r$ is a triangle. Otherwise, there is a chord in $C_r$ incident to $r$, and it forms a subcycle of $C$ which contains both $r$ and $s$. On the other hand, the vertices of $C_s$ are pairwise adjacent due to the minimality of $C$. Hence, $C_s$ is also a triangle. Otherwise, there exists a chord in $C_s$ incident to $s$, and it forms a subcycle of $C$ which contains both $r$ and $s$. Therefore, $C$ is a bad induced diamond in $G_j$.
 
Next, consider the case when $r\notin V(C)$. From this point forward, we may assume that every smallest violated cycle in $G_j$ does not contain $r$. Otherwise, we are back in the first case again. With this additional assumption, non-adjacency within $C$ implies non-adjacency with $r$, as shown by the following claim.

\begin{claim}\label{clm:root}
For any pair of vertices $u,v\in V(C)$ such that $uv\notin E_j$, we have $ru\notin E_j$ or $rv\notin E_j$.
\end{claim}

\begin{proof}
For the purpose of contradiction, suppose $ru,rv\in E_j$. Let $s\in V(C)$ such that $rs\notin E_j$. Let $P_{su}$ and $P_{sv}$ denote the edge-disjoint $s$-$u$ and $s$-$v$ paths in $C$ respectively. Let $u'$ and $v'$ be the closest vertex to $s$ on $P_{su}$ and $P_{sv}$ respectively such that $ru',rv'\in E_j$ (see Figure \ref{fig:nonadjacent} for an example). Without loss of generality, let $w(ru')\geq w(rv')$. Denote $P_{su'}$ and $P_{sv'}$ as the $s$-$u'$ and $s$-$v'$ subpaths of $P_{su}$ and $P_{sv}$ respectively. Now, consider the cycle $D:=P_{su'}+P_{sv'}+ru'+rv'$. Observe that it contains $r$ and is no bigger than $C$. Furthermore, it does not contain a chord incident to $r$ by our choice of $u'$ and $v'$. To arrive at a contradiction, it is left to show that $D$ is well-covered, as this would imply $D$ is violated. Suppose for a contradiction, that $D$ is not well-covered. Then, there exists a chord $g$ in $D$ such that $w(g)<w(ru')$ and $w(g)<w(h)$ for some $h\in E(D)$, where $ru'$ and $h$ lie in different subcycles of $D$ formed by $g$. This chord exists because $C$ is well-covered but $D$ is not. Let $C_1$ and $C_2$ denote the subcycles of $C$ formed by $g$, where $h\in E(C_2)$. Note that $C_1$ is well-covered because $w(g)<w(h)$. Moreover, we also have $w(g)\geq w(e)$ for all $e\in E(C_1)$ because $C$ is well-covered. Let $w(g)=w_\ell$ for some $\ell<j$. Then, $C_1$ is still present in $G_\ell$ but not $ru'$. Since $C_1$ also contains $u,v$ and $uv\notin E_\ell$, it is a violated cycle in $G_\ell$. However, this is a contradiction because there are no violated cycles in $G_\ell$.
\end{proof}

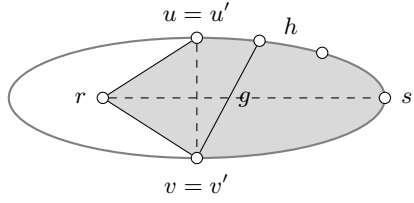
\begin{figure}
\centering
\def \hrad {2.5}
\def \vrad {0.8}
\begin{tikzpicture}
	\begin{scope}
		\clip (0,0) ellipse (\hrad cm and \vrad cm);
		\path [fill=gray!30] (-\hrad/2,0) -- (0,\vrad) -- (\hrad,\vrad) -- (\hrad,-\vrad) -- (0,-\vrad) -- (-\hrad/2,0);
	\end{scope}

	\draw[thick,color=black!50] (0,0) ellipse (\hrad cm and \vrad cm);

	\node[node,label=left:\footnotesize{$r$}] (r) at (-\hrad/2,0) {};
	\node[node,label=right:\footnotesize{$s$}] (s) at (\hrad,0) {};
	\node[node,label=above:\footnotesize{$u=u'$}] (u) at (0,\vrad) {};
	\node[node,label=below:\footnotesize{$v=v'$}] (v) at (0,-\vrad) {};
	\node[node] (v1) at (\hrad/3,0.95*\vrad) {};
	\node[node] (v2) at (\hrad/1.5,0.75*\vrad) {};

	\node [font=\footnotesize] (h) at (0.5*\hrad,1.2*\vrad) {$h$};

	\path[dashed edge] (r) -- (s);
	\path[dashed edge] (u) -- (v);
	\path[edge] (r) -- (u);
	\path[edge] (r) -- (v);
	\path[edge] (v1) -- node [weight,right] {\footnotesize{$g$}} (v);

\end{tikzpicture}
\caption{The ellipse represents the violated cycle $C$ in Claim \ref{clm:root}. The shaded region highlights the violated cycle $D$. The dashed edges indicate non-adjacency in $G_j$. In this example, $u=u'$ and $v=v'$.}
\label{fig:nonadjacent}
\end{figure}

The remaining proof proceeds in a similar fashion to the first case. Let $u,v\in V(C)$ such that $uv\notin E_j$. By the claim above, we know that $ru\notin E_j$ or $rv\notin E_j$. For any chord $f$ in $C$, observe that $u$ and $v$ lie in different subcycles of $C$ formed by $f$. This is because the subcycles are well-covered, so the one which contains both $u$ and $v$ will contradict the minimality of $C$. Now, let $g$ be a chord of $C$. Let $C_u$ and $C_v$ denote the subcycles of $C$ formed by $g$ where $u\in V(C_u)$ and $v\in V(C_v)$. The vertices of $C_u$ are pairwise adjacent due to the minimality of $C$. Thus, $C_u$ is a triangle. Otherwise, there exists a chord in $C_u$ incident to $u$, and it forms a subcycle of $C$ which contains both $u$ and $v$. By an analogous argument, $C_v$ is also a triangle. Therefore, $C$ is a bad induced diamond in $G_j$.
\end{proof}

Notice that we have proven something stronger. Namely, if $G_j$ contains a violated cycle, then there exists an $i\leq j$ such that $G_i$ contains a bad hole or a bad induced diamond. Moreover, as mentioned earlier, bad holes and hamiltonian cycles of bad induced diamonds are violated cycles themselves. Thus, we obtain the following corollary.

\begin{corollary}
\label{cor:violated}
There are no bad holes or bad induced diamonds in $G_i$ for all $i<k$ if and only if there are no violated cycles in $G_i$ for all $i<k$.
\end{corollary}

\section{Candidate edges and expensive neighborhood}
\label{sec:candidate_edges}
In the previous section, we have shown that violated cycles are an obstruction to submodularity. In light of this fact, we now focus on graphs which do not contain violated cycles. For the sake of brevity, we will use $(\star)$ to denote the following property:

\begin{center}
\textit{There are no violated cycles in $G_i$ for all $i<k$.}
\end{center}

The goal of this section is to study the behaviour of $f_{uv}$ assuming $(\star)$ holds. As a first step, the following lemma sheds light on how a minimum spanning tree changes under vertex removal.

\begin{lemma}\label{lem:reconnect}
Assume $(\star)$ holds. Let $T$ be a minimum spanning tree of $G[S]$ where $r\in S\subseteq V$. For any $s\neq r$, there exists a minimum spanning tree of $G[S\setminus s]$ which contains $E(T\setminus s)$ and additionally, only uses edges from $G[N_T(s)\cup r]$.
\end{lemma}

\begin{proof}
Pick a vertex $s\in S\setminus r$. By Lemma \ref{lem:reuse}, there exists a minimum spanning tree of $G[S\setminus s]$ which contains $E(T\setminus s)$. Let $T'$ be such a tree which uses the most edges from $G[N_T(s)\cup r]$. We will show that $T'$ is our desired tree. For the purpose of contradiction, suppose $T'$ has an edge $uv$ where $uv\notin E(T)$ and $u\notin N_T(s)\cup r$. Note that $u$ and $v$ lie in different components of $T\setminus s$. Let $P_{su}$ and $P_{sv}$ denote the unique $s$-$u$ and $s$-$v$ paths in $T$ respectively. Then, $C:=P_{su}\cup P_{sv} \cup uv$ is a well-covered cycle in $G_i$ where $w(uv)=w_i$. Let $u'$ and $v'$ be the vertices adjacent to $s$ in $P_{su}$ and $P_{sv}$ respectively. By our choice of $T'$, $w(u'v')>w(uv)$. Since $uv$ is the most expensive edge in $C$, the vertices of $C$ are not pairwise adjacent in $G_i$. So they are adjacent to $r$ in $G_i$. However, adding $ru'$ or $rv'$ to $T'$ creates a fundamental cycle which uses the edge $uv$. Swapping it with $uv$ creates another minimum spanning tree of $G[S\setminus s]$ which contains $E(T\setminus s)$ and uses more edges from $G[N_T(s)\cup r]$. We have arrived at a contradiction.
\end{proof}

Given a pair of vertices $u,v\in N$ where $w(uv)=w_i$, the following definition distinguishes the neighbours of $u,v$ in $G$ from the neighbours of $u,v$ in $G_i$.

\begin{definition}
For an edge $uv\in E$, if $w(uv)=w_i$, the \emph{expensive neighborhood} of $uv$ is defined as
\[\hat{N}(uv) := N_k(uv)\setminus N_i(uv).\]
\end{definition}

In other words, the expensive neighborhood of an edge $uv$ is the set of vertices $s\notin \set{u,v}$ such that $\max\set{w(su),w(sv)}>w(uv)$. It turns out that the function $f_{uv}$ always returns zero when evaluated on a set which does not lie entirely in the expensive neighborhood of $uv$.

\begin{lemma}\label{lem:zero}
Assume $(\star)$ holds. Let $u,v\in N$ and $S\in \mathcal{S}_{uv}$. If $S\not\subseteq \hat{N}(uv)$, then $f_{uv}(S)=0$.
\end{lemma}

\begin{proof}
Let $T$ be a minimum spanning tree of $G[S\cup \set{u,v}]$. First, we show that we can assume $uv\notin E(T)$. Since $S\not\subseteq\hat{N}(uv)$, there exists a vertex $s\in S$ such that $\max\set{w(su),w(sv)}\leq w(uv)$. If $uv\in E(T)$, then by rooting $T$ at $s$, $u$ is either a child or a parent of $v$. Adding $su$ to $T$ in the former and $sv$ in the latter creates a fundamental cycle which contains $uv$. Thus, we can replace $uv$ with this new edge to obtain the desired tree. Now, by Lemma \ref{lem:reconnect}, there exists a minimum spanning tree $T'$ of $G[S\cup v]$ which contains $E(T\setminus u)$ and additionally, only uses edges from $G[N_T(u)\cup r]$. Since $v\notin N_T(u)\cup r$, the neighborhood of $v$ is identical in both trees, i.e.~$N_T(v)=N_{T'}(v)$.

Consider the forest $T\setminus v$. Let $p\in N_T(v)$ such that $p$ and $r$ lie in the same component of $T\setminus v$ (see Figure \ref{fig:trees} for an example). Note that $p=r$ if $r\in N_T(v)$. We claim that $p$ and $r$ also lie in the same component of the forest $T'\setminus v$. We may assume that $p\neq r$, as otherwise the claim is trivially true. Moreover, we may assume that $u$ lies on the unique $p$-$r$ path in $T$. Otherwise, we are done because the same path is present in $T'\setminus v$. Let $C_r$ denote the component of $T\setminus v$ which contains $p,r$ and $u$. By Lemma \ref{lem:reconnect}, the endpoints of every edge in $E(T')\setminus E(T\setminus u)$ lie in $C_r$. This proves the claim. 

Using Lemma \ref{lem:reconnect}, we can construct a minimum spanning tree of $G[S\cup u]$ by deleting $v$ from $T$ and adding a set of edges $F$ from $G[N_T(v)\cup r]$. Note that $pr\notin F$ as $p$ and $r$ lie in the same component of $T\setminus v$. Since $p$ and $r$ also lie in the same component of $T'\setminus v$ and $N_T(v)=N_{T'}(v)$, deleting $v$ from $T'$ and adding $F$ creates a minimum spanning tree of $G[S]$. Thus, we get
\begin{align*}
  f_{uv}(S) &= \textsf{mst}(S\cup u) + \textsf{mst}(S\cup v) - \textsf{mst}(S) - \textsf{mst}(S\cup\set{u,v}) \\
  &= \Big(\textsf{mst}(S\cup u) - w(T)\Big) - \Big(\textsf{mst}(S) - w(T')\Big) \\
  &= \Big(w(F) - w(\delta_T(v))\Big) - \Big(w(F) - w(\delta_{T'}(v))\Big) = 0
\end{align*}
as desired.
\end{proof}

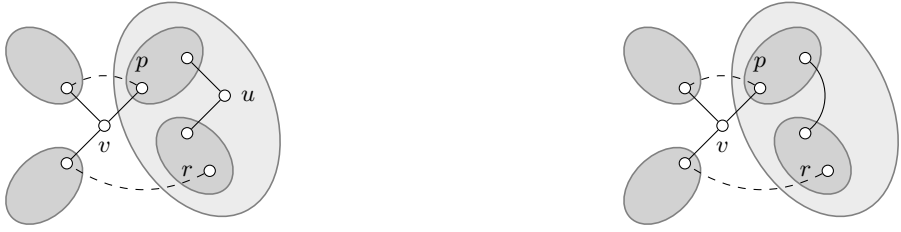
\begin{figure}
\def \shrad {0.7}
\def \svrad {0.4}
\def \bhrad {1.5}
\def \bvrad {1.2}
\centering
\begin{minipage}{0.49\textwidth}
\centering
\begin{tikzpicture}
	\draw[rotate around={25:(1,0)},semithick,color=black!50,fill=gray!15] (1.3,0.1) ellipse (1cm and 1.5cm);

	\draw[rotate around={-45:(-0.8,0.8)},semithick,color=black!50,fill=gray!35] (-0.8,0.8) ellipse (0.6 cm and 0.4 cm);
	\draw[rotate around={45:(-0.8,-0.8)},semithick,color=black!50,fill=gray!35] (-0.8,-0.8) ellipse (0.6 cm and 0.4 cm);
	\draw[rotate around={45:(0.8,0.8)},semithick,color=black!50,fill=gray!35] (0.8,0.8) ellipse (0.6 cm and 0.4 cm);
	\draw[rotate around={45:(1.2,-0.4)},semithick,color=black!50,fill=gray!35] (1.2,-0.4) ellipse (0.4 cm and 0.6 cm);

	\node[node,label=below:\footnotesize{$v$}] (v) at (0,0) {};
	\node[node,label=above:\footnotesize{$p$}] (p) at (0.5,0.5) {};
	\node[node,label=right:\footnotesize{$u$}] (u) at (1.6,0.4) {};
	\node[node,label=left:\footnotesize{$r$}] (r) at (1.4,-0.6) {};
	\node[node] (s1) at (-0.5,0.5) {};
	\node[node] (s2) at (-0.5,-0.5) {};
	\node[node] (s3) at (1.1,0.9) {};
	\node[node] (s4) at (1.1,-0.1) {};

	\path[edge] (v) -- (p);
	\path[edge] (v) -- (s1);
	\path[edge] (v) -- (s2);
	\path[edge] (u) -- (s3);
	\path[edge] (u) -- (s4);
	\path[dashed edge] (p) to [bend right=30] (s1);
	\path[dashed edge] (r) to [bend left=30] (s2);
\end{tikzpicture}
\end{minipage}
\begin{minipage}{0.49\textwidth}
\centering
\begin{tikzpicture}
	\draw[rotate around={25:(1,0)},semithick,color=black!50,fill=gray!15] (1.3,0.1) ellipse (1cm and 1.5cm);

	\draw[rotate around={-45:(-0.8,0.8)},semithick,color=black!50,fill=gray!35] (-0.8,0.8) ellipse (0.6 cm and 0.4 cm);
	\draw[rotate around={45:(-0.8,-0.8)},semithick,color=black!50,fill=gray!35] (-0.8,-0.8) ellipse (0.6 cm and 0.4 cm);
	\draw[rotate around={45:(0.8,0.8)},semithick,color=black!50,fill=gray!35] (0.8,0.8) ellipse (0.6 cm and 0.4 cm);
	\draw[rotate around={45:(1.2,-0.4)},semithick,color=black!50,fill=gray!35] (1.2,-0.4) ellipse (0.4 cm and 0.6 cm);

	\node[node,label=below:\footnotesize{$v$}] (v) at (0,0) {};
	\node[node,label=above:\footnotesize{$p$}] (p) at (0.5,0.5) {};
	\node[node,label=left:\footnotesize{$r$}] (r) at (1.4,-0.6) {};
	\node[node] (s1) at (-0.5,0.5) {};
	\node[node] (s2) at (-0.5,-0.5) {};
	\node[node] (s3) at (1.1,0.9) {};
	\node[node] (s4) at (1.1,-0.1) {};

	\path[edge] (v) -- (p);
	\path[edge] (v) -- (s1);
	\path[edge] (v) -- (s2);
	\path (s3) edge [bend left=50] (s4);
	\path[dashed edge] (p) to [bend right=30] (s1);
	\path[dashed edge] (r) to [bend left=30] (s2);
\end{tikzpicture}
\end{minipage}
\caption{The left image depicts an example of the minimum spanning tree $T$ in $G[S\cup\set{u,v}]$. The right image depicts an example of the minimum spanning tree $T'$ in $G[S\cup v]$. The solid edges belong to the trees while dashed edges belong to the edge set $F$.}
\label{fig:trees}
\end{figure}

We can now focus solely on vertex sets which lie entirely in the expensive neighborhood of $uv$. Observe that if $r\notin \hat{N}(uv)$, then $S\not\subseteq \hat{N}(uv)$ for all $S\in\mathcal{S}_{uv}$. Thus, we do not have to check these edges as $f_{uv}(S)=0$ for all $S\in \mathcal{S}_{uv}$ by the previous lemma. This motivates the following definition:

\begin{definition}
An edge $uv\in E$ is called a \emph{candidate edge} if $r\in \hat{N}(uv)$.
\end{definition}

With a mild assumption, we can show that the function $f_{uv}$ is inclusion-wise nonincreasing in the expensive neighborhood of $uv$.

\begin{lemma}\label{lem:monotone}
Assume $(\star)$ holds and $f_{xy}(\hat{N}(xy))	\geq 0$ for every candidate edge $xy$. Let $uv$ be a candidate edge and $S\in \mathcal{S}_{uv}$ such that $S\subseteq \hat{N}(uv)$. For any $s\neq r$, $f_{uv}(S)\leq f_{uv}(S\setminus s)$.
\end{lemma}

\begin{proof}
Pick a vertex $s\in S\setminus r$. Without loss of generality, assume $w(su)\geq w(sv)$. Then, $w(su)>w(uv)$ because $s
\in \hat{N}(uv)$. However, these two inequalities also imply that $v\notin \hat{N}(su)$. It follows that the set $(S\setminus s)\cup v$ is not contained in the expensive neighborhood of $su$. By Lemma \ref{lem:zero},
\[0 = f_{su}((S\setminus s)\cup v) = \textsf{mst}(S\cup v) + \textsf{mst}((S\setminus s)\cup \set{u,v}) - \textsf{mst}((S\setminus s)\cup v) - \textsf{mst}(S\cup\set{u,v}).\]
Rearranging yields
\begin{equation}\label{eqn:one}
	\textsf{mst}(S\cup v) - \textsf{mst}(S\cup\set{u,v}) = \textsf{mst}((S\setminus s)\cup v) - \textsf{mst}((S\setminus s)\cup\set{u,v}).
\end{equation}
Since $uv$ is a candidate edge, let $w(uv)=w_i$ for some $i<k$. We will proceed by induction on $i$. For the base case $i=k-1$, we have $w_{k-1} = w(uv)<w(su)=w_k$. Since $\hat{N}(su)=\emptyset$, the set $S\setminus s$ is not contained in the expensive neighborhood of $su$ because $r\in S\setminus s$. By Lemma \ref{lem:zero},
\[0 = f_{su}(S\setminus s) = \textsf{mst}(S) + \textsf{mst}((S\setminus s)\cup u) - \textsf{mst}(S\setminus s) - \textsf{mst}(S\cup u).\]
Rearranging yields
\begin{equation}\label{eqn:two}
	\textsf{mst}(S\cup u) - \textsf{mst}(S) = \textsf{mst}((S\setminus s)\cup u) - \textsf{mst}(S\setminus s).
\end{equation}
Adding (\ref{eqn:one}) and (\ref{eqn:two}) gives $f_{uv}(S)=f_{uv}(S\setminus s)$. Now, suppose the lemma is true for all $i\geq j$ for some $j<k$. For the inductive step, let $w(uv)=w_{j-1}$. We may assume that $S\setminus s \subseteq \hat{N}(su)$, as otherwise we obtain equality again. This implies that $su$ is a candidate edge because $r\in S\setminus s$. Since $w(su)>w(uv)=w_{j-1}$, we get
\[0 \leq f_{su}(\hat{N}(su)) \leq f_{su}(S\setminus s) = \textsf{mst}(S) + \textsf{mst}((S\setminus s)\cup u) - \textsf{mst}(S\setminus s) - \textsf{mst}(S\cup u)\]
where the first inequality is due to our assumption while the second inequality is due to the induction hypothesis. Then, rearranging and adding it to (\ref{eqn:one}) yields $f_{uv}(S)\leq f_{uv}(S\setminus s)$ as desired.
\end{proof}

\section{Characterization of submodularity}
\label{sec:characterization}
We are finally ready to give an efficient characterization of submodular spanning tree games. 

\begin{theorem}
\label{thm:main}
The spanning tree game on $G$ is submodular if and only if:
\begin{enumerate}
  \item[(i)] There are no violated cycles in $G_i$ for all $i<k$.
  \item[(ii)] For every candidate edge $uv$, $f_{uv}(\hat{N}(uv))\geq 0$.
\end{enumerate}
Furthermore, these conditions can be verified in polynomial time.
\end{theorem}

\begin{proof}
For necessity, assume the game is submodular. Then, Condition $(i)$ is satisfied by Lemma \ref{lem:violated} while Condition $(ii)$ is satisfied trivially. For sufficiency, assume Conditions $(i)$ and $(ii)$ hold. Let $u,v\in N$ and $S\in \mathcal{S}_{uv}$. If $S\not\subseteq \hat{N}(uv)$, then $f_{uv}(S)=0$ by Lemma \ref{lem:zero}. On the other hand, if $S\subseteq \hat{N}(uv)$, then $uv$ is a candidate edge. By Lemma \ref{lem:monotone}, 	
\[f_{uv}(S)\geq f_{uv}(\hat{N}(uv))\geq 0.\]
Therefore, the game is submodular. 

Checking Condition $(ii)$ can clearly be done in polynomial time. Using Corollary \ref{cor:violated}, verifying Condition $(i)$ reduces to searching for bad holes and bad induced diamonds in each $G_i$, which can be done efficiently. To look for bad holes, one could check if there exists a hole through a given vertex $v$ for all $v\in N$ where $rv\notin E_i$. To look for bad induced diamonds, a naive implementation would involve examining all vertex subsets of size 4, which still runs in polynomial time.
\end{proof}

\section{$S$-wide spanning trees}
\label{sec:wide}
In this section, we answer another question posed in \cite{journal/networks/KobayashiO14}. There the authors stated a necessary condition for submodularity of the spanning tree game, and left open whether their condition can be verified in polynomial time. We here show that this is indeed the case. While this is not that relevant anymore in order to characterize submodularity (since we have provided Theorem~\ref{thm:main}), it shows a nice connection with matroid intersection, that might still be of interest.

\begin{theorem}[Theorem 1.2 in \cite{journal/networks/KobayashiO14}]
\label{thm:wide}
If the spanning tree game on $G$ is submodular, then every minimum spanning tree $T$ of $G$ possesses the following two properties:
\begin{enumerate}
	\item[(a)] It holds that $w(rv) \geq  w(ru)$ for every vertex $v \in N$ and every vertex $u \in N$ on the (unique) path connecting $r$ and $v$ in $T$. 
	\item[(b)] For any pair of vertices $u,v \in N$ where $w(uv) < w(rv)$, the cycle obtained by adding $uv$ to $T$ does not contain $r$.
\end{enumerate}
\end{theorem}

We show that checking Conditions $(a)$ and $(b)$ can be reduced to computing \emph{$S$-wide} spanning trees, which is defined as follows.
Let $H=(W,F)$ be an undirected graph with edge-costs $c\in \R^{\size{F}}$ and a designated root $r_1\in W$. Let $S:=\set{s_1,\dots,s_k}\subseteq W\setminus r_1$ be a given set of \emph{terminals}. We say that a spanning tree $T$ is \emph{$S$-wide} if every component of $T\setminus r_1$ contains at most one terminal. Equivalently, $T$ is \emph{$S$-wide} if for every $i,j\in [k]$ where $i\neq j$, the unique $s_i-s_j$ path on $T$ contains the root. We are interested in the following problem: 

\begin{center}
\textit{Given $(H,c,r_1,S)$, compute an $S$-wide spanning tree of minimum cost.}
\end{center}

Before solving the problem above, let us explain how one can use it to efficiently test Conditions $(a)$ and $(b)$ of Theorem~\ref{thm:wide}.
To check Condition $(a)$, consider all pairs of vertices $u,v \in N$ with $w(rv) <  w(ru)$, and do the following. Set $H:=G$, $c_e = w(e) \; \forall e \in E$, $r_1 := u$ and $S := \{v,r\}$. Compute an $S$-wide spanning tree $T^*$ of minimum cost in $H$. If the cost of $T^*$ is equal to the weight of a minimum spanning tree of $G$, then $T^*$ is a minimum spanning tree of $G$ violating Condition $(a)$, since $u$ is on the unique path
from $r$ to $v$ in $T^*$. It is not difficult to see that this procedure will eventually find a minimum spanning tree violating Condition $(a)$, if one exists.

Condition $(b)$ can be checked in a similar way. Consider all ordered pairs of vertices $u,v \in N$ with $w(uv) <  w(rv)$, and do the following. Set $H:=G$, $c_e = w(e) \; \forall e \in E$, $r_1 := r$ and $S := \{u,v\}$. Compute an $S$-wide spanning tree $T^*$ of minimum cost in $H$. If the cost of $T^*$ is equal to the weight of a minimum spanning tree of $G$, then $T^*$ is a minimum spanning tree of $G$ violating Condition $(b)$, since adding $uv$ to $T^*$ yields a cycle containing $r$. Once again, it is not difficult to see that this procedure will eventually find a minimum spanning tree violating Condition $(b)$, if one exists.

We will now demonstrate how to compute an $S$-wide spanning tree of minimum cost using matroid intersection.
Let $(H,c,r_1,S)$ be a given instance. Without loss of generality, we may assume that there are no edges between any pair of terminals, as every $S$-wide tree does not use them.  We also assume $k\geq 2$, otherwise this reduces trivially to computing an arbitrary minimum spanning tree. 

First, construct an auxiliary graph $H'=(W',F')$ from $H$ as follows. Create $k-1$ copies of the root vertex $r_2,\dots,r_k$, including its incident edges. Now, let $H'_1$ denote the (multi)graph obtained from $H'$ by identifying $r_1,\dots,r_k$ into a single vertex $r$. Similarly, let $H'_2$ denote the (multi)graph obtained from $H'$ by identifying $s_1,\dots,s_k$ into a single vertex $s$. Note that parallel edges are kept. See Figure \ref{fig:auxgraph} for an example. 

\begin{figure}[ht]
\def\dist{1.5}
\begin{minipage}{0.33\textwidth}
\centering
\begin{tikzpicture}[node distance=\dist cm, inner sep=2.5pt, minimum size=2.5pt, auto]
  \node [node,label=above:\footnotesize{$r_1$}] (r1) {};
  \node [node,label=above:\footnotesize{$r_2$}] (r2) [right of=r1] {};
  \node [node] (u) [below of=r1] {};
  \node [node] (v) [below of=r2] {};
  \node [node] (t1) [below of=u,label=below:\footnotesize{$s_1$}] {};
  \node [node] (t2) [below of=v,label=below:\footnotesize{$s_2$}] {};

  \path (r1) edge (u);
  \path (r1) edge (v);
  \path (r1) edge [bend right=20] (t1);
  \path (r1) edge (t2);
  \path (r2) edge (u);
  \path (r2) edge (v);
  \path (r2) edge (t1);
  \path (r2) edge [bend left=20] (t2);

  \path (u) edge (v);
  \path (u) edge (t1);
  \path (u) edge (t2);

  \path (v) edge (t1);
  \path (v) edge (t2);
\end{tikzpicture}
\end{minipage}
\begin{minipage}{0.32\textwidth}
\centering
\begin{tikzpicture}[node distance=\dist cm, inner sep=2.5pt, minimum size=2.5pt, auto]
  \node [node,label=above:\footnotesize{$r$}] (r) {};
  \node [node] (u) [below left of=r] {};
  \node [node] (v) [below right of=r] {};
  \node [node] (t1) [below of=u,label=below:\footnotesize{$s_1$}] {};
  \node [node] (t2) [below of=v,label=below:\footnotesize{$s_2$}] {};

  \path (r) edge (u);
  \path (r) edge (v);
  \path (r) edge (t1);
  \path (r) edge (t2);
  \path (r) edge [bend right=20] (u);
  \path (r) edge [bend left=20] (v);
  \path (r) edge [bend right=15] (t1);
  \path (r) edge [bend left=15] (t2);

  \path (u) edge (v);
  \path (u) edge (t1);
  \path (u) edge (t2);

  \path (v) edge (t1);
  \path (v) edge (t2);


\end{tikzpicture}
\end{minipage}
\begin{minipage}{0.33\textwidth}
\centering
\begin{tikzpicture}[node distance=\dist cm, inner sep=2.5pt, minimum size=2.5pt, auto]
  \node [node,label=below:\footnotesize{$s$}] (t) {};
  \node [node] (u) [above left of=t] {};
  \node [node] (v) [above right of=t] {};
  \node [node] (r1) [above of=u,label=above:\footnotesize{$r_1$}] {};
  \node [node] (r2) [above of=v,label=above:\footnotesize{$r_2$}] {};

  \path (t) edge (u);
  \path (t) edge (v);
  \path (t) edge (r1);
  \path (t) edge (r2);
  \path (t) edge [bend left=20] (u);
  \path (t) edge [bend right=20] (v);
  \path (t) edge [bend left=15] (r1);
  \path (t) edge [bend right=15] (r2);

  \path (u) edge (v);
  \path (u) edge (r1);
  \path (u) edge (r2);

  \path (v) edge (r1);
  \path (v) edge (r2);
\end{tikzpicture}
\end{minipage}
\caption{The graphs $H'$, $H'_1$ and $H'_2$ respectively for the input graph $K_5$ with 2 terminals.}
\label{fig:auxgraph}
\end{figure}
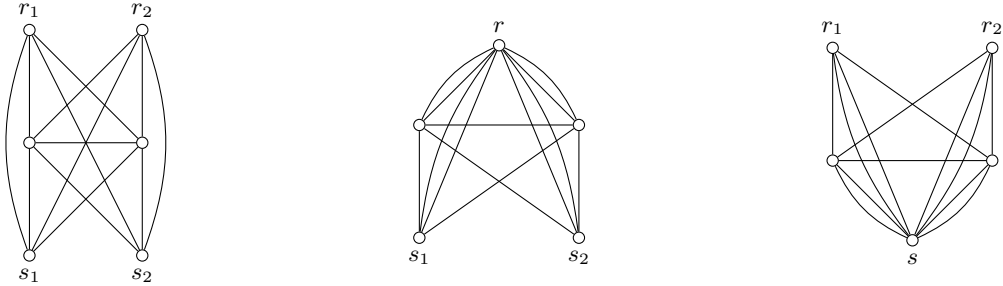

Next, consider the graphic matroids on $H'_1$ and $H'_2$, denoted $M_1 = (F',\mathcal{I}_1)$ and $M_2 = (F',\mathcal{I}_2)$ respectively. We would like to establish a correspondence between $S$-wide spanning trees in $H$ and common bases of $M_1$ and $M_2$. Let $\mathcal{T}$ be the set of $S$-wide spanning trees in $H$. Moreover, let $\mathcal{B}_1$ and $\mathcal{B}_2$ be the set of bases of $M_1$ and $M_2$ respectively. Construct the function $g$ as 
\[g(J) := \set{f(e):e\in J}\]
where $f:F'\rightarrow F$ is defined as
\[f(uv):=
\begin{cases}
	r_1v, &\text{ if }u\in \set{r_1,\dots,r_k}\\
	ur_1, &\text{ if }v\in \set{r_1,\dots,r_k}\\
	uv, &\text{ otherwise.}
\end{cases}\]

\begin{lemma}
For every $J\in \mathcal{B}_1\cap \mathcal{B}_2$, we have $g(J)\in \mathcal{T}$.
\end{lemma}

\begin{proof}
Let $J\in \mathcal{B}_1\cap \mathcal{B}_2$. Then, it is a spanning tree in $H'_1$ and $H'_2$. Note that $g(J)$ is also a spanning tree in $H'_1$. Since $H=H'_1\setminus \set{r_iv\in F':2\leq i \leq k}$ and $g(J)$ is present in $H$, it follows that $g(J)$ is a spanning tree in $H$. It is left to show that $g(J)$ is $S$-wide. For the purpose of contradiction, suppose there exist $i,j\in [k]$ such that $i\neq j$ and the $s_i-s_j$ path in $g(J)$ does not contain the root. Among all such paths, pick the shortest one (in terms of number of edges) and call it $P$. Since $P$ is the shortest, none of its intermediate vertices are terminals. Note that the edges of $P$ belong to $J$, and they form a cycle in $H'_2$, which is a contradiction.
\end{proof}

\begin{lemma}
For every $T\in \mathcal{T}$, there exists a $J\in \mathcal{B}_1\cap \mathcal{B}_2$ such that $g(J)=T$.
\end{lemma}

\begin{proof}
Let $T$ be an $S$-wide spanning tree in $H$. For every $i\in [k]$, let $r_1v_i$ be the first edge of the $r_1-s_i$ path on $T$. Replace each $r_1v_i$ with $r_iv_i$ and call the resulting edge set $J$. Then, $J$ consists of $k$ subtrees in $H'$, each of which contains $r_i$ and $s_i$. Identifying the $r_i$'s makes $J$ a spanning tree in $H'_1$, while identifying the $s_i$'s makes $J$ a spanning tree in $H'_2$. Thus, $J\in \mathcal{B}_1\cap\mathcal{B}_2$. Moreover, $g(J)=T$.
\end{proof}

The last two lemmas imply that $g$ is a surjective mapping from $\mathcal{B}_1\cap \mathcal{B}_2$ to $\mathcal{T}$. Therefore, we can efficiently compute a minimum cost $S$-wide spanning tree by computing a minimum cost common basis of $M_1$ and $M_2$.

\paragraph{Acknowledgements.}
We would like to thank Kanstantsin Pashkovich for suggesting the connection with matroids in Section \ref{sec:wide}. This work was supported by the NSERC Discovery Grant Program and an Early Researcher Award by the Province of Ontario.

\bibliographystyle{abbrv}
\bibliography{references}

\begin{thebibliography}{10}

\bibitem{journal/networks/Bird76}
C.~G. Bird.
\newblock On cost allocation for a spanning tree: A game theoretic approach.
\newblock {\em Networks}, 6(4):335--350, 1976.

\bibitem{book/synthesis/ChalkiadakisEW11}
G.~Chalkiadakis, E.~Elkind, and M.~Wooldridge.
\newblock {\em Computational Aspects of Cooperative Game Theory}.
\newblock Synthesis Lectures on Artificial Intelligence and Machine Learning.
  Morgan {\&} Claypool Publishers, 2011.

\bibitem{journal/networks/ClausK73}
A.~Claus and D.~J. Kleitman.
\newblock Cost allocation for a spanning tree.
\newblock {\em Networks}, 3(4):289--304, 1973.

\bibitem{journal/ijgt/FaigleKFH97}
U.~Faigle, W.~Kern, S.~P. Fekete, and W.~Hochst{\"a}ttler.
\newblock On the complexity of testing membership in the core of min-cost
  spanning tree games.
\newblock {\em International Journal of Game Theory}, 26(3):361--366, 1997.

\bibitem{journal/ijgt/FaigleKK98}
U.~Faigle, W.~Kern, and J.~Kuipers.
\newblock Note computing the nucleolus of min-cost spanning tree games is
  {NP}-hard.
\newblock {\em International Journal of Game Theory}, 27(3):443--450, 1998.

\bibitem{journal/mp/GranotH81}
D.~Granot and G.~Huberman.
\newblock Minimum cost spanning tree games.
\newblock {\em Math. Program.}, 21(1):1--18, 1981.

\bibitem{journal/siamadm/GranotH82}
D.~Granot and G.~Huberman.
\newblock The relationship between convex games and minimum cost spanning tree
  games: A case for permutationally convex games.
\newblock {\em SIAM Journal on Algebraic Discrete Methods}, 3(3):288--292,
  1982.

\bibitem{book/algorithms/GroetschelLS88}
M.~Gr{\"o}tschel, L.~Lov{\'a}sz, and A.~Schrijver.
\newblock {\em Geometric Algorithms and Combinatorial Optimization}, volume~2
  of {\em Algorithms and Combinatorics}.
\newblock Springer, 1993.

\bibitem{journal/networks/KobayashiO14}
M.~Kobayashi and Y.~Okamoto.
\newblock Submodularity of minimum-cost spanning tree games.
\newblock {\em Networks}, 63(3):231--238, 2014.

\bibitem{report/Kuipers96}
J.~Kuipers.
\newblock A polynomial time algorithm for computing the nucleolus of convex
  games.
\newblock {\em Report M 96--12, Maastricht University}, 1996.

\bibitem{journal/ijgt/MaschlerPS71}
M.~Maschler, B.~Peleg, and L.~S. Shapley.
\newblock The kernel and bargaining set for convex games.
\newblock {\em International Journal of Game Theory}, 1(1):73--93, 1971.

\bibitem{journal/mmor/Okamoto03}
Y.~Okamoto.
\newblock Submodularity of some classes of the combinatorial optimization
  games.
\newblock {\em Mathematical Methods of Operations Research}, 58(1):131--139,
  2003.

\bibitem{journal/ijgt/Shapley71}
L.~S. Shapley.
\newblock Cores of convex games.
\newblock {\em International Journal of Game Theory}, 1(1):11--26, 1971.

\bibitem{journal/geb/Trudeau12}
C.~Trudeau.
\newblock A new stable and more responsive cost sharing solution for minimum
  cost spanning tree problems.
\newblock {\em Games and Economic Behavior}, 75(1):402 -- 412, 2012.

\bibitem{journal/ijgt/vandenNouwelandB91}
A.~van~den Nouweland and P.~Borm.
\newblock On the convexity of communication games.
\newblock {\em International Journal of Game Theory}, 19(4):421--430, 1991.

\end{thebibliography}

\end{document}